\documentclass[12pt]{amsart}

\usepackage{amsmath}

\usepackage{amsfonts,amssymb,amsthm,enumitem}
\usepackage[margin=1.2in]{geometry}
\usepackage{mathtools}
\usepackage{nicefrac,setspace}
\usepackage{xcolor}
\usepackage{amsmath}

\usepackage{algorithm}
\usepackage{algpseudocode}
\usepackage{enumitem}

\usepackage{enumitem}
\usepackage{hyperref}       
\usepackage{url}            
\usepackage{booktabs}       
\usepackage{amsfonts}       
\usepackage{nicefrac}       
\usepackage{microtype}      
\usepackage{comment}   
\usepackage{mathtools}

\makeatletter

\newcommand\gam{\kern-.4ex\raisebox{.40ex}{$\gamma$}\kern-.40ex $_{N_0}$}

\newcommand\gamsq{\kern-.4ex\raisebox{.40ex}{$\gamma$}\kern-.40ex $_{N_0}^{\hspace{1mm}2}$}

\newcommand*\bigcdot{\mathpalette\bigcdot@{.5}}
\newcommand*\bigcdot@[2]{\mathbin{\vcenter{\hbox{\scalebox{#2}{$\m@th#1\odot$}}}}}
\makeatother

\newtheorem{theorem}{Theorem}

\newtheorem*{theorem*}{Theorem}
\newtheorem{remark}{Remark}
\newtheorem*{claim*}{Claim}
\newtheorem*{note*}{Note}
\newtheorem*{remark*}{Remark}
\newtheorem*{lemma*}{Lemma}

\newtheorem{lemma}[theorem]{Lemma}

\newtheorem*{fact*}{Fact}
\newtheorem{corollary}[theorem]{Corollary}

\newcommand{\N}{\mathbb{N}}

\newcommand{\F}{\mathbb{F}}
\newcommand{\E}{\mathop \mathbb{E}}

\title{From Bit to Block: Decoding on Erasure Channels}

\newcommand{\HP}[1]{\textcolor{green!40!black}{HP: #1}}
\newcommand{\OS}[1]{\textcolor{blue}{OS: #1}}

\author{Henry D. Pfister$^*$}
\address{$^*$ Duke University, henry.pfister@duke.edu}

\author{Oscar Sprumont$^{\dagger}$}
\address{$^\dagger$ University of Washington, osprum@cs.washington.edu}

\author{Gilles Z{\'{e}}mor$^{\ddagger}$}
\address{$^\ddagger$ Universit\'e de Bordeaux and Institut universitaire de France, zemor@math.u-bordeaux.fr}

\usepackage[foot]{amsaddr}

\begin{document}

\begin{abstract}
    We provide a general framework for bounding the block error threshold of a linear code $C\subseteq \F_2^N$ over the erasure channel in terms of its bit error threshold. Our approach relies on understanding the minimum support weight of any $r$-dimensional subcode of $C$, for all small values of $r$. As a proof of concept, we use our machinery to obtain a new proof of the celebrated result that Reed--Muller codes achieve capacity on the erasure channel with respect to block error probability.
\end{abstract}

\maketitle

\section{Introduction}

We will be interested in the performance of linear codes over noisy communication channels. Ever since Shannon's seminal paper \cite{shannon1948entropy}, information theorists have spent considerable time and effort designing codes that \emph{achieve capacity}, i.e. codes $C\subseteq\F_2^N$ that approach the optimal tradeoff between their rate $R:=\frac{\log_2|C|}{N}$ and the amount of noise they can tolerate.

Shannon's probabilistic argument \cite{shannon1948entropy} established that uniform
random codes achieve capacity on both the binary symmetric channel and the erasure channel. The first explicit families of codes to provably achieve capacity were only obtained decades later, when Forney introduced concatenated codes \cite{forney1966concatenated}. More recently, Arikan showed that polar codes - which have both a deterministic construction and efficient encoding and decoding algorithms - 
achieve capacity on all memoryless channels \cite{arikan2009polar}. This brought renewed attention to the closely related Reed--Muller codes, which were shown to achieve capacity on the binary symmetric channel and the erasure channel in \cite{abbe2015RMlowrate,kudekar2016erasure,kudekar2016bitblock,reeves2021bitcapacity,abbe2023camelliacodes,abbe2023rmcapacityBSC}.

The proof that Reed--Muller codes of constant rate achieve capacity on the erasure channel \cite{kudekar2016erasure,kudekar2016bitblock} relies heavily on the double transitivity of their automorphism group. In fact, Kudekar, Kumar, Mondelli, Pfister, Şaşoğlu and Urbanke showed in \cite{kudekar2016erasure} that \emph{any} doubly transitive linear code achieves capacity under bit-MAP decoding. A very natural question is then to understand whether or not the bit and block error thresholds of doubly transitive codes are asymptotically equal. (By the result of \cite{kudekar2016erasure}, this is equivalent to asking whether or not every doubly transitive code achieves capacity on the erasure channel under block-MAP decoding.)

\subsection{Main Results}
Our main contribution is to develop a framework for bounding the gap between a code's bit and block error thresholds on the erasure channel. Let $C\subseteq \F_2^N$ be any linear code. We show that if, for all integers $r<r_0$ (where $r_0$ depends on the sharpness of the bit threshold of $C$), any $r$-dimensional subcode of $C$ has support\footnote{The support of a subcode $S\subseteq C$ is the set of indices $j\in[N]$ where at least one codeword $c\in S$ has $c_j=1$.} growing faster than $r\log N$, then the block error threshold of $C$ is close to its bit error threshold. An informal version of our main result is stated below. For the formal version, see Theorem \ref{bittoblock}.

\begin{theorem}[Informal]\label{bittoblockinformal}
    Consider any linear code $C\subseteq\F_2^N$ and suppose that a uniform random codeword $c\in C$ is sent over the erasure channel with erasure probability $p\in[0,1]$. Let $\delta\in[0,1]$ be such that, for every $i\in [N]$, the probability we fail to decode the $i^\textnormal{th}$ coordinate of $c$ is bounded by
    \begin{align*}
        \Pr_{p\textnormal{-erasures}}\Big[ \textnormal{cannot recover  }c_i\Big]\leq \delta.
    \end{align*}
    Suppose additionally that for all $r=1,2,\dotsc, \sqrt{\delta}N$, the support of any $r$-dimensional subcode of $C$ has size $\omega(r\log N).$ Then there exists $p'=p-o(1)$ such that the probability we fail to decode the codeword $c$ from erasures with probability $p'$ is bounded by
    \begin{align*}
         \Pr_{p'\textnormal{-erasures}}\Big[ \textnormal{cannot recover }c \Big]\leq \sqrt{\delta}+o(1).
    \end{align*}
\end{theorem}  
As a proof of concept, we show in Section \ref{sectionrm} that Reed--Muller codes of constant rate satisfy the conditions of Theorem \ref{bittoblockinformal}. Our analysis relies on the work of Wei, who obtained in \cite{wei1991RMsubcodes} exact expressions for the minimum support weight of any $r$-dimensional subcode of a Reed--Muller code. Bounding these expressions appropriately and combining them with Theorem \ref{bittoblockinformal} gives an alternative proof that Reed--Muller codes achieve capacity on the erasure channel under block-MAP decoding, a fact that was shown by Kudekar, Kumar, Mondelli, Pfister, Şaşoğlu and Urbanke in \cite{kudekar2016erasure,kudekar2016bitblock}. See Theorem \ref{rmcapacity}.

We note that one key argument in the proof of \cite{kudekar2016bitblock} is the claim that the low-weight codewords of Reed-Muller codes do not contribute much to the decoding error probability, and thus one only needs to worry about codewords of high weight. In order to prove this claim, \cite{kudekar2016bitblock} relies on a very careful analysis of Reed--Muller codes' weight enumerator, which may be difficult to generalize to other codes. In contrast, our approach only requires somewhat loose bounds on the minimum support weight of any subcode. For $r<N^{1-o(1)}$, we need any $r$-dimensional subcode to have support weight larger than $r\log N$; for constant-rate Reed--Muller codes, it turns out that any such subcode has support weight larger than $r2^{(\log N)^{0.99}}$ (see Corollary \ref{ratiorm}).
\subsection{Techniques}
Our proof of Theorem \ref{bittoblockinformal} is based on a well-known work of Tillich and Z{\'{e}}mor \cite{tillich2000sharptransition}, who established sharp threshold results for the symmetric and erasure channels. For a linear code $C\subseteq\F_2^N$ and a string $x\in\{0,1\}^N$, we define
\begin{align}\label{defnsc}
    S_C(x):=\Big\{c\in C:c_i\leq x_i\textnormal{ for all } i\in[N] \Big\}
\end{align}
to be the set of codewords that are indistinguishable from the $0$-vector once you erase all coordinates $i\in[N]$ where $x_i=1$. Note that $S_C(x)$ is a subcode of $C$. Tillich and Z{\'{e}}mor showed in \cite{tillich2000sharptransition} that for any fixed $r\in\{1,2,\dotsc,\textnormal{dim }C\}$, the function
\begin{align}\label{defnfr}
    f_r(p):=\Pr_{x\sim p}\Big[\textnormal{dim }S_C(x)\geq r\Big]
\end{align}
transitions rapidly from $\approx 0$ to $\approx 1$ as a function of $p$, where $x\sim p$ denotes a $p$-noisy random string $x\in\{0,1\}^N.$ They also showed that for any $r$, the curves $f_{r}(p)$ and $f_{r+1}(p)$ stay within a distance of about $\frac{1}{\sqrt{d_r(C)}}$ from each other, where $d_r(C)$ denotes the minimum support weight of any $r$-dimensional subcode of $C$. 

In this work, we first strengthen their result to show that the curves $f_{r}(p)$ and $f_{r+1}(p)$ in fact stay within a distance of about $\frac{1}{d_r(C)}$ from each other (see the proof of Theorem \ref{straightshot}). We then leverage the fact that every doubly transitive code achieves capacity under bit-MAP decoding to prove that for some $r_0=N^{1-o(1)}$, the function $f_{r_0}$ corresponding to any doubly transitive code $C\subseteq\F_2^N$ satisfies 
\begin{align*}
f_{r_0}\Big(1-\textnormal{rate}(C)-o(1)\Big)=o(1) .
\end{align*}
Since the distance between the curves $f_r(p)$ and $f_{r+1}(p)$ can be bounded in terms of the minimum support weight $d_r(C)$, we are then able to prove (see Theorem \ref{thm2transitive}) that the function $f_1(p)$ corresponding to any doubly transitive code $C\subseteq\F_2^{N}$ with large enough minimum support weights satisfies
\begin{align}\label{endeqintro}
f_{1}\Big(1-\textnormal{rate}(C)-o(1)\Big)=o(1) .
\end{align} 
But the left hand side above is exactly the probability that a sent codeword $c\in C$ can be uniquely recovered from random erasures of probability $1-\textnormal{rate}(C)-o(1)$, so equation (\ref{endeqintro}) shows that any doubly transitive linear code $C\subseteq\F_2^N$ with large enough minimum support weights $\{d_r(C)\}$ achieves capacity on the erasure channel under block-MAP decoding. See Theorem \ref{bittoblock} for the formal statement.

\section{Preliminaries}\label{sectionnotation}
\subsection{Notation and Conventions}
We will be interested in the behavior of linear codes $C\subseteq\F_2^N$ over the erasure channel BEC$_p$, for some $p\in[0,1].$ Throughout this paper, we use $\ln X$ to denote the natural logarithm of a real number $X>0$ and use $\log X$ to denote the logarithm of $X$ in base 2. We denote by $[N]$ the set of integers $\{1,2,\dotsc,N\}$ and use the shorthand 
\begin{align*}
    x\sim p
\end{align*}
to signify that the random variable $x\in\F_2^N$ has independent
Bernoulli coordinates
\begin{align*}
    x_i=\begin{cases}
1 & \text{with probability $p$},\\
0 & \text{otherwise}.
\end{cases}
\end{align*}
We will think of the coordinates $j\in[N]$ where $x_j=1$ as the coordinates that are erased by the channel and will call $x$ the \emph{erasure pattern}. Formally, when a sender sends some codeword $c\in C$ through the channel BEC$_p$, an erasure pattern $x\sim p$ is sampled and the receiver receives a vector $y\in \{0,1,*\}^N$ with coordinates
\begin{align*}
    y_i=\begin{cases}
c_i & \text{if $x_i=0$},\\
* & \text{otherwise}.
\end{cases}
\end{align*}
A notion that will be important in our analysis is the notion of \emph{covered codeword}. For any erasure pattern $x\in\{0,1\}^N$, we say that a codeword $c\in C$ is \emph{covered} by $x$ if $x_i\geq c_i$ for all $i\in[N]$, and denote this by
\begin{align*}
    x\succ c.
\end{align*}
We note that if the erasure pattern $x$ covers some codeword $c\neq 0$, then it is impossible for the receiver to distinguish between a sent message $c'\in C$ and the codeword $c+c'.$ We denote the subcode of $C$ covered by an erasure pattern $x$ by
\begin{align*}
    S_C(x):=\Big\{c\in C:x\succ c \Big\}.
\end{align*}
A sent codeword $c\in C$ can be recovered uniquely from an erasure pattern $x\in\F_2^N$ if and only if $S_C(x)=\{0\}$. Similarly, for any coordinate $i\in [N]$, the bit $c_i$ can be recovered if and only if $i\notin \textnormal{supp}\big(S_C(x)\big)$, where for any 
subcode $S\subseteq C$ we define  
\begin{align}\label{defsupport}
    \textnormal{supp}\big(S\big):=\Big\{ j\in[N]:\exists c\in S \textnormal{ with }c_j=1  \Big\}
\end{align}
to be the \emph{support} of $S$. For every $r\leq \textnormal{dim } C$, we define the function $g_r\colon\{0,1\}^N\rightarrow \{0,1\}$ to be
\begin{align}\label{defngr}
    g_r(x):=\begin{cases}
1    & \text{if $\textnormal{dim }S_C(x) \geq r$},\\
0 & \text{otherwise}.
\end{cases}
\end{align}
For any fixed linear code $C\subseteq\F_2^N$, we will be interested in the expected value of its corresponding function $g_r$,
\begin{align*}
	f_r(p):=\E_{x\sim p}\big[g_r(x) \big],
\end{align*}
as well as the inverse map $\theta_r\colon[0,1]\rightarrow [0,1]$,
\begin{align*}
    \theta_r(\alpha):=f_r^{-1}(\alpha).
\end{align*}
The main quantity that will allow us to control the behavior of these two functions is the minimum support size of any $r$-dimensional subcode of $C$,
 \begin{align}\label{defdr}
    d_r(C):=\min_{\substack{S\subseteq C\\\textnormal{dim($S)$}=r}}\big|\textnormal{supp($S)$}\big|.
\end{align}

One explicit family of codes we will be interested in is the family of Reed--Muller codes. We denote by $\mathsf{RM}(n,d)$ the Reed--Muller code with
$n$ variables and degree $d$. The codewords of $\mathsf{RM}(n,d)$ are the evaluation vectors (over all points in $\F_2^n$) of all polynomials of degree $\leq d$ in $n$ variables. The Reed--Muller code $\mathsf{RM}(n,d)$ has dimension $\binom{n}{\leq d}$ (see for instance
page 5 of \cite{abbe2021survey}).

\subsection{Known Results}
Our work makes use of several well-known results, which we state below. 
First, we will need the following bound on the bit error decay of any doubly transitive linear code. It is due to Kudekar, Kumar, Mondelli, Pfister, Şaşoğlu and Urbanke. 
\begin{theorem}[follows from \cite{kudekar2016erasure}, Lemma 34 and Theorem 19]\label{bitcapacity}
    Let $C\subseteq \F_2^N$ be a doubly transitive code and fix any index $i\in[N]$. Define $p^*\in[0,1]$ to be the noise parameter at which $$\Pr_{x\sim p^*}\Big[i\in \textnormal{supp}\big(S_C(x)\big)\Big]=\frac{1}{2}.$$ 
    Then for all $p\leq p^*$, we have
    \begin{align*}
        \Pr_{x\sim p}\Big[i\in \textnormal{supp}\big(S_C(x)\big)\Big]\leq e^{-(p^*-p)\log (N-1)}.
    \end{align*}
\end{theorem}

We will also need the following result of Wei, which states that the subcodes of Reed--Muller codes with smallest supports are Reed--Muller codes of smaller degrees on fewer variables. 
\begin{theorem}[follows from \cite{wei1991RMsubcodes}, Theorem 7] \label{wei}
    For every $t\leq d\leq n$, a $\binom{n-t}{\leq d-t}$-dimensional subcode of $\mathsf{RM}(n,d)$ with smallest support is the subcode $S_t\subseteq\mathsf{RM}(n,d)$, which contains the evaluation vectors of all polynomials of the form
$$p(x_{t+1},x_{t+2},\dotsc,x_{n})\prod_{i=1}^{t}x_i,$$
    for $p(x_{t+1},x_{t+2},\dotsc,x_{n})$ any polynomial of degree $\leq d-t$ in the variables $x_{t+1},x_{t+2},\dotsc,x_{n}.$ 
\end{theorem}
Finally, we will need some tools to understand the function $g_r$ defined in (\ref{defngr}). For any function $g\colon\F_2^N\rightarrow\{0,1\}$, we define the function
\begin{align}\label{defnhg}
    h_g(x):=\begin{cases}
\Big| \big\{i\in[N]:g(x+e_i)=0   \big\} \Big| & \text{if $g(x)=1$},\\
0 & \text{otherwise}
\end{cases}
\end{align}
and its minimal nonzero value
\begin{align}\label{defnvg}
    \nu_g:=\min\big\{h_g(x):x\in\F_2^n \textnormal{ such that } h_g(x)\neq 0\big\}.
\end{align} 
The quantities (\ref{defnhg}) and 
(\ref{defnvg}) were introduced by Margulis \cite{margulis1974transition} and Russo \cite{russo1982transition} to prove sharp transition results for the expectation of monotone\footnote{A Boolean function $g:\{0,1\}^N\rightarrow\{0,1\}$ is said to be \emph{monotone} if for all $x,y\in\{0,1\}^N$ with $x_i\geq y_i$ for all $i\in[N]$, we have $g(x)\geq g(y).$} Boolean functions. The following lemma, which relates the derivative of the expectation of any monotone function $g$ to the expectation of its corresponding function $h_g$, is often called the \emph{Margulis-Russo Lemma}.
\begin{lemma}[Margulis-Russo Lemma, \cite{margulis1974transition,russo1982transition}]\label{margulisrussolemma}
    For any monotone Boolean function $g\colon\{0,1\}^N\rightarrow\{0,1\}$ and any $p\in[0,1]$, we have
\begin{align*}
        \frac{d}{dp}\E_{x\sim p}[g(x)]=\frac{1}{p}\cdot\E_{x\sim p}[h_g(x)].
    \end{align*}
\end{lemma}

Tillich and Z{\'{e}}mor showed that the quantity $\nu_{g_r}$ defined in (\ref{defngr}) and (\ref{defnvg}) can be bounded by the quantity $d_r(C)$ defined in (\ref{defdr}).

\begin{lemma}[follows from \cite{tillich2000sharptransition}, page 476]\label{erasuretillich}
Consider any linear code $C\subseteq \F_2^N.$ For every $1\leq r\leq \textnormal{dim }C$, the corresponding function $g_r$ satisfies
\begin{align*}
        \nu_{g_r}\geq d_r(C).
    \end{align*}
    Additionally, for $1\leq r\leq \textnormal{dim }C-1$ we have
    \begin{align*}
        \Pr_{x\sim p}\Big[ h_{g_r}(x)\neq 0 \Big]=\Pr_{x\sim p}\Big[g_{r}(x)=1\Big]-\Pr_{x\sim p}\Big[g_{r+1}(x)=1\Big].
    \end{align*}
\end{lemma}

See \cite{tillich2000sharptransition} for the proof. Intuitively, the two statements in Lemma \ref{erasuretillich} follow from the following two facts:
\begin{enumerate}
    \item If an index $i\in [N]$ is in the support of the subcode covered by the erasure pattern $x\in\{0,1\}^N$, then half of the codewords $c\in S_C(x)$ have $c_i=1.$ Thus, if you remove the erasure symbol on the $i^\textnormal{th}$ coordinate of the erasure pattern $x$, then you only cover half as many codewords as you did before (and thus, the dimension of the covered subcode has gone down by 1).
    \item Removing the erasure symbol on any coordinate $i\in[N]$ of an erasure pattern $x$ can reduce the number of covered codewords by at most half (for every $i\in [N]$, at least half of the codewords $c\in S_C(x)$ have $c_i=0$). Thus, removing the erasure symbol on any coordinate $i\in[N]$ of an erasure pattern $x$ can decrease the dimension of the covered subcode by at most $1.$
\end{enumerate}

\section{General Linear Codes}

In this section, we will provide general conditions under which the bit and block error thresholds of an arbitrary linear code $C\subseteq\F_2^N$ are close to one another. We start by bounding the distance between the curves $\theta_1(\alpha)$ and $\theta_{N_0}(\alpha)$, for any $N_0>1.$

\begin{theorem}\label{straightshot}
Fix any linear code $C\subseteq\F_2^N$ and any integer $N_0\leq\textnormal{dim }C$. Then, letting \gam{}$:=\sqrt{\sum_{r=1}^{N_0-1}\frac{1}{d_r(C)}}$, we have that for any $$0\leq\alpha\leq 1- \textnormal{\gam},$$ the functions $\theta_1$ and $\theta_{N_0}$ associated with $C$ satisfy
    \begin{align*}
         \theta_{1}\big(\alpha+\textnormal{\gam} \big)\geq \theta_{N_0}(\alpha)  -\textnormal{\gam}.
    \end{align*}
\end{theorem}

\begin{proof}
    Consider any $r\leq \textnormal{dim }C.$
    By Lemmas \ref{margulisrussolemma} and \ref{erasuretillich}, the function $f_r$ associated with our code $C$ satisfies
    \begin{align}\label{integraleq}
        \frac{d}{dp}f_r(p)&\geq  \Pr_{x\sim p}\Big[ h_{g_r}(x)\neq 0 \Big]\cdot d_r(C)\nonumber\\
        &=\Big(f_r(p)-f_{r+1}(p)\Big)d_r(C).
    \end{align} 
    Since each $\theta_r$ is the inverse function of $f_r$, the area between the curves $\theta_r(\alpha)$ and $\theta_{r+1}(\alpha)$ is the same as the area between the curves $f_r(p)$ and $f_{r+1}(p)$. Thus, we find that
\begin{align}\label{boundeachtheta}
     \int_0^1 \Big( \theta_{r+1} (\alpha) - \theta_{r} (\alpha) \Big) d\alpha &= \int_0^1 \Big( f_r (p) - f_{r+1} (p) \Big) dp\nonumber \\
     &\leq \frac{1}{d_r(C)},     
\end{align}
where the last line follows from applying the Fundamental Theorem of Calculus to (\ref{integraleq}). Taking $\alpha$ uniformly at random from the interval $[0,1]$ and applying (\ref{boundeachtheta}) with $r=1,2,\dotsc,N_0-1$ 
then gives
\begin{align*}
    \E_{\alpha\in[0,1]}\Big[\theta_{N_0}(\alpha)-\theta_1(\alpha)\Big]&= \int_0^1 \Big( \theta_{N_0}(\alpha)-\theta_1(\alpha)\Big)d\alpha\\
    &= \sum_{r=1}^{N_0-1}\int_0^1 \Big( \theta_{r+1} (\alpha) - \theta_{r} (\alpha) \Big) d\alpha\\
    &\leq \sum_{r=1}^{N_0-1}\frac{1}{d_r(C)}\\
    &=\textnormal{\gamsq}.
\end{align*}
By Markov's inequality, we must then have that for $\alpha$ uniformly random in $[0,1]$,
\begin{align*}
    \Pr_{\alpha\in[0,1]}\Big[\theta_{N_0}(\alpha)-\theta_1(\alpha)>\textnormal{\gam}\Big]
    &\leq \textnormal{\gam}.
\end{align*}
In particular, we see that, for any $\alpha\in \big[0,1-\textnormal{\gam}\big]$, there must be some $\alpha'\in\big[\alpha,\alpha +\textnormal{\gam}\big]$ such that
\begin{align*}
    \theta_{N_0}(\alpha')-\theta_1(\alpha')\leq \textnormal{\gam}.
\end{align*}
Since the functions $\theta_r$ are increasing, it follows that
\begin{align*}
    \theta_{N_0}(\alpha)-\theta_1\big(\alpha+\textnormal{\gam}\big)\leq \textnormal{\gam}.
\end{align*}
\end{proof}

We are now ready to prove our main result. The following theorem is a formal version of Theorem \ref{bittoblockinformal}.
\begin{theorem}\label{bittoblock}
    Consider any linear code $C\subseteq\F_2^N$ with $N\geq 10$. Let $p,\delta\in[0,1]$ be such that $\textnormal{dim }C\geq \sqrt{\delta}N$ and
    \begin{align*}
        \Pr_{x\sim p}\Big[i\in \textnormal{supp}\big(S_C(x)\big)\Big]\leq \delta
    \end{align*}
    for every $i\in [N]$. Define
    \begin{align*}
        \Delta:=\min_{r=1,2,\dotsc, \sqrt{\delta} N}\Big\{\frac{d_r(C)}{r}  \Big\}.
    \end{align*}
    Then we have
    \begin{align*}
         \Pr_{x\sim p-\sqrt{\frac{\log N}{\Delta}}}\Big[ \exists c\in C:x\succ c \Big]\leq \sqrt{\delta}+\sqrt{\frac{\log N}{\Delta}}.
    \end{align*}
\end{theorem}

\begin{proof}
Note that we may assume that
\begin{align}\label{deltafine}
    \sqrt{\delta}+\sqrt{\frac{\log N}{\Delta}}\leq 1,
\end{align}
as otherwise the claim is trivial. 
By linearity of expectation, we have
    \begin{align*}
        \E_{x\sim p}\Big[\textnormal{supp}\big(S_C(x)\big)\Big]\leq \delta N.
    \end{align*}
Since the dimension of a subspace can at most be as large as its support, Markov's inequality then gives us
\begin{align*}
    \Pr_{x\sim p}\Big[\textnormal{dim}\big(S_C(x)\big)\geq\sqrt{\delta} N\Big]
    &\leq\frac{ \mathop{\E}_{x\sim p}\Big[\textnormal{supp}(S_C(x))\Big]}{\sqrt{\delta} N }\\
    &\leq  \sqrt{\delta},
\end{align*}
or equivalently,
\begin{align}\label{startpoint}
    \theta_{\sqrt{\delta} N}(\sqrt{\delta})\geq p.
\end{align}
On the other hand, by definition of $\Delta$ we have
\begin{align}\label{logbound}
    \sum_{r=1}^{\sqrt{\delta} N}\frac{1}{d_r(C)}\nonumber
    &\leq \frac{1}{\Delta}\sum_{r=1}^{\sqrt{\delta} N} \frac{1}{r}\nonumber\\
    &\leq \frac{\log N}{\Delta},
\end{align}
where in the last line we used the facts that $N\geq 10$ and that $\sum_{r=1}^{N'} \frac{1}{r}\leq \ln(N')+1$ for all $N'\geq 1$. By inequalities (\ref{deltafine}) and (\ref{logbound}), the conditions of Theorem \ref{straightshot} are satisfied for $N_0=\sqrt{\delta}N$ and $\alpha=\sqrt{\delta}$. Applying Theorem \ref{straightshot}, we then have
\begin{align*}
    \theta_{1}\left(\sqrt{\delta}+\sqrt{\frac{\log N}{\Delta}}\right)&\geq \theta_{1}\left(\sqrt{\delta}+\sqrt{\sum_{r=1}^{\sqrt{\delta} N}\frac{1}{d_r(C)}}\right)\\
    &\geq\theta_{\sqrt{\delta} N}(\sqrt{\delta})-\sqrt{\sum_{r=1}^{\sqrt{\delta} N}\frac{1}{d_r(C)}}\\
    &\geq\theta_{\sqrt{\delta} N}(\sqrt{\delta})-\sqrt{\frac{\log N}{\Delta}},
\end{align*}
where in the first and third lines we used the inequality (\ref{logbound}). By equation (\ref{startpoint}), we get
\begin{align*}
    \theta_{1}\left(\sqrt{\delta}+\sqrt{\frac{\log N}{\Delta}}\right)
    &\geq p-\sqrt{\frac{\log N}{\Delta}}.
\end{align*}
Since $\theta_1$ is the inverse function of $f_1$, this inequality is equivalent to our desired claim.
\end{proof}

\section{Doubly Transitive Codes}
In this section, we apply our Theorem \ref{bittoblock} to doubly transitive codes to bound the gap between the bit and block error thresholds of any doubly transitive linear code.
\begin{theorem}\label{thm2transitive}
    Let $C\subseteq \F_2^N$ be any doubly transitive linear code with $N\geq 10$.  Define $p^*\in[0,1]$ to be the noise parameter at which $\Pr_{x\sim p^*}\big[1\in S_C(x)\big]=\frac{1}{2}$, and fix any $p\leq p^*$. Suppose $\textnormal{dim }C\geq Ne^{-\frac{p^*-p}{2}\log (N-1)}$. Then defining
    \begin{align*}
        \Delta:=\min_{r=1,2,\dotsc, Ne^{-\frac{p^*-p}{2}\log (N-1)}}\Big\{\frac{d_r(C)}{r}  \Big\},
    \end{align*}
    we have
    \begin{align*}
        \Pr_{x\sim p-\sqrt{\frac{\log N}{\Delta}}}\Big[ \exists c\in C:x\succ c\Big]\leq \sqrt{\frac{\log N}{\Delta}}+e^{-\frac{p^*-p}{2}\log (N-1)}.
    \end{align*}
\end{theorem}
\begin{proof}

By Theorem \ref{bitcapacity}, we have that for all $i\in[N]$,
\begin{align*}
        \Pr_{x\sim p}\Big[i\in \textnormal{supp}\big(S_C(x)\big)\Big]\leq e^{-(p^*-p)\log (N-1)}.
    \end{align*}
 Applying Theorem  \ref{bittoblock} with $\delta=e^{-(p^*-p)\log (N-1)}$, we then get 
    \begin{align*}
         \Pr_{x\sim p-\sqrt{\frac{\log N}{\Delta}}}\Big[ \exists c\in C:x\succ c \Big]\leq \sqrt{\frac{\log N}{\Delta}}+e^{-\frac{p^*-p}{2}\log (N-1)}.
    \end{align*}
    
\end{proof}

\section{Reed--Muller Codes}\label{sectionrm}
We will now use Theorem \ref{thm2transitive} to show that Reed--Muller codes achieve capacity on the erasure channel. For any $t\leq d\leq n$, consider the subcode $S_t\subseteq\mathsf{RM}(n,d)$ defined in 
Theorem \ref{wei}. Recall that Wei showed in \cite{wei1991RMsubcodes} that $S_t$ is a subcode of minimal support for its dimension (see Theorem \ref{wei}). 
We bound its dimension and support explicitly below.

\begin{lemma}\label{rmbounds}
    For every $n\in\N$, every $d\leq \frac{n}{2}+\sqrt{n\log n}$ and every $t\in[5\sqrt{n\log n}, d]$, the corresponding subcode $S_t\subseteq\mathsf{RM}(n,d)$ satisfies
\begin{align*}
    \big|\textnormal{supp}(S_t)\big|\geq 2^{n-t},
\end{align*}
\begin{align*}
    \textnormal{dim}(S_t)\leq 2^{n-t}\cdot 2^{-\frac{1}{4} \left(\frac{t^2}{n}-\frac{t^3}{n^2} \right)}.
\end{align*}
\end{lemma}
\begin{proof}
    The first statement follows from the fact that the evaluation vector of the monomial $\prod_{i=1}^t x_i$ is in $S_t$, and the fact that $\prod_{i=1}^t x_i$ evaluates to $1$ on all points $x\in\F_2^n$ with $x_1=x_2=\cdots=x_t=1.$ For the second statement, we compute
    \begin{align*}
        \textnormal{dim}(S_t)&=\binom{n-t}{\leq d-t}\\
        &\leq \binom{n-t}{\leq\frac{n}{2}-\frac{4t}{5}},
    \end{align*}
    where in the second line we used the fact that $d\leq \frac{n}{2}+\sqrt{n\log n}\leq \frac{n}{2}+\frac{t}{5}$. Since $\binom{m}{\leq s}\leq 2^{h(\frac{s}{m})m}$ for all $s\leq \frac{m}{2}$ (see for example \cite{binomialbound}, Theorem 3.1) we then get
    \begin{align*}
        \textnormal{dim}(S_t)&\leq 2^{h\big(\frac{\frac{1}{2}-\frac{4t}{5n}}{1-\frac{t}{n}}\big)(n-t)}\\
        &\leq 2^{h\big(\frac{1-\frac{3t}{5n}}{2}\big)(n-t)},
    \end{align*}
where in the second line we used the Taylor expansion $\frac{1}{1-x}=1+\sum_{j=1}^\infty x^j$ to bound $\frac{\frac{1}{2}-\frac{4x}{5}}{1-x}=\frac{1}{2}-\frac{3x}{10(1-x)}\leq\frac{1}{2}-\frac{3x}{10}$ for all $x\in[0,1)$. Applying the inequality $h(\frac{1-x}{2})\leq 1-\frac{x^2}{2\ln2}$, we then get
    \begin{align*}
        \textnormal{dim}(S_t)&\leq 2^{(1-\frac{t^2}{4n^2})(n-t)},
    \end{align*}
    as desired.
\end{proof}
As a corollary of Lemma \ref{rmbounds}, we get the following bound on the minimum size of the support of any $r$-dimensional subcode of a Reed--Muller code.

\begin{corollary}\label{ratiorm}
    For every $n$ large enough, every $d\leq \frac{n}{2}+\sqrt{n\log n}$ and every $\epsilon\in\left[6\sqrt{\frac{\log n}{n}},\frac{1}{2}\right] $, we have
    \begin{align*}
        \frac{d_r(\mathsf{RM}(n,d))}{r}\geq 2^{\frac{\epsilon^2 n}{10}}
    \end{align*}
    for all $r\leq 2^{n-\epsilon n}$.
\end{corollary}
\begin{proof}
We first note that since the minimum distance of $\mathsf{RM}(n,d)$ is $2^{n-d}$, it will suffice to prove our claim for every $r\in[2^{n-d-\frac{\epsilon^2 n}{10}}, 2^{n-\epsilon n}]$. Consider any such $r$, and let the integer $t\in[\epsilon n-\frac{\epsilon^2 n}{10},d-1]$ be such that 
\begin{align}\label{pickt}
    2^{n-t-1-\frac{\epsilon^2 n}{10}}\leq r\leq2^{n-t-\frac{\epsilon^2 n}{10}}.
\end{align}

Note that by Lemma \ref{rmbounds} and our theorem's condition on $\epsilon$, for any $t\in[\epsilon n-\frac{\epsilon^2 n}{10},d-1]$ we have
\begin{align*}
    \textnormal{dim}(S_t)\leq 2^{n-t}\cdot 2^{-\frac{n}{4}\cdot \frac{t^2}{n^2}\left(1-\frac{t}{n} \right)}.
\end{align*}
But the function $x^2(1-x)$ is increasing over $[0,\frac{2}{3}]$. Since $t\leq \frac{n}{2}+\sqrt{n\log n}\leq \frac{2n}{3}$ for all $n$ large enough, we then get (because $\frac{t}{n}\geq \epsilon-\frac{\epsilon^2}{10}$)
\begin{align*}
    \textnormal{dim}(S_t)&\leq 2^{n-t}\cdot 2^{-\frac{n}{4} \left(\epsilon-\frac{\epsilon^2}{10}\right)^2\left(1-\epsilon+\frac{\epsilon^2}{10}\right)}\\
    &\leq 2^{n-t}\cdot 2^{-\frac{n}{4} \left(\frac{19\epsilon}{20}\right)^2\cdot\frac{1}{2}}\\
    &\leq 2^{n-t}\cdot 2^{-\frac{\epsilon^2 n}{10}-1},
\end{align*}
where in the second line we used the fact that $\epsilon\leq \frac{1}{2}$.
Combining this with the leftmost inequality of (\ref{pickt}), we get $r\geq 2^{n-t-1-\frac{\epsilon^2 n}{10}}\geq \textnormal{dim}(S_{t})$. By Theorem~\ref{wei} and Lemma \ref{rmbounds}, we must then have 
\begin{align*}
    d_r\big( \mathsf{RM}(n,d)\big)  &\geq \big|\textnormal{supp}(S_{t})\big|\\
    &\geq 2^{n-t}.
\end{align*}
Combining this inequality with the right-hand side of (\ref{pickt}), we get
\begin{align*}
        \frac{d_r\big( \mathsf{RM}(n,d)\big)}{r}&\geq 2^{\frac{\epsilon^2 n}{10}}.
    \end{align*}
\end{proof}

We are now ready to prove that the bit and block error thresholds of Reed--Muller codes are asymptotically equal. Since every doubly transitive code achieves capacity under bit-MAP decoding \cite{kudekar2016erasure}, this implies that Reed--Muller codes achieve capacity under block-MAP decoding.

\begin{theorem}\label{rmcapacity}
     For every $n$ large enough, every $d\in[\frac{n}{2}-\sqrt{n\log n}, \frac{n}{2}+\sqrt{n\log n}]$ and every $\epsilon\geq20\sqrt{\frac{\log n}{n}} $, the Reed--Muller code $\mathsf{RM}(n,d)$ satisfies
\begin{align*}
    \Pr_{x\sim p^*-\epsilon}\Big[\exists c\in C:x\succ c\Big]\leq 2^{-\frac{\epsilon^2 n}{100}},
\end{align*}
where $p^*\in[0,1]$ is such that $\Pr_{x\sim p^*}\big[1\in \textnormal{supp}(S_C(x))\big]=\frac{1}{2}$.
\end{theorem}

\begin{proof}
    Note that we may assume that $\epsilon\leq 1$. Note also that
    \begin{align*}
        \textnormal{dim}\big(\mathsf{RM}(n,d)\big)&=\binom{n}{\leq d}\\
        &\geq \binom{n}{\leq \frac{n}{2}-\sqrt{n\log n}}\\
        &\geq \frac{1}{\sqrt{2n}}\cdot 2^{h\left( \frac{1}{2}-\sqrt{\frac{\log n}{n}} \right)n},
    \end{align*}
    where in the third line we used the inequality $\binom{n}{ k}\geq\frac{1}{\sqrt{2n}}\cdot 2^{h(k/n)n}$ (see for example \cite{1977book}, page 309, Lemma 7). Applying the inequality $h(\frac{1-x}{2})\geq 1-x^2$, we then get 

    \begin{align*}
        \textnormal{dim}\big(\mathsf{RM}(n,d)\big)&\geq \frac{1}{\sqrt{2n}}2^{(1-4\frac{\log n}{n})n}\\
        &\geq 2^ne^{-\frac{\epsilon(n-1)}{6}}.
    \end{align*}

    By Theorem \ref{thm2transitive}, for every $p\leq p^*-\frac{\epsilon}{3}$, we then have 
    \begin{align}\label{boundp}
        \Pr_{x\sim p-\sqrt{\frac{n}{\Delta}}}\Big[ \exists c\in C:x\succ c\Big]&\leq \sqrt{\frac{n}{\Delta}}+e^{-\frac{p^*-p}{2}\log (2^n-1)}\nonumber\\
        &\leq \sqrt{\frac{n}{\Delta}}+e^{-\frac{p^*-p}{2}(n-1)}
    \end{align}
    for
    \begin{align*}
        \Delta:=\min_{r=1,2,\dotsc, Ne^{-\frac{p^*-p}{2}(n-1)}}\Big\{\frac{d_r\big(\mathsf{RM}(n,d)\big)}{r}  \Big\}.
    \end{align*}
    Letting $p=p^*-\frac{2\epsilon}{3}\cdot\frac{n}{n-1}$ and applying Corollary \ref{ratiorm} with parameter $\frac{\epsilon}{3\ln2}$ (the conditions of Corollary \ref{ratiorm} are satisfied by our theorem's condition on $\epsilon$), we get
    \begin{align*}
        \Delta\geq2^{\frac{\epsilon^2n}{90(\ln2)^2}}\geq n2^{\frac{\epsilon^2n}{50}+2},
    \end{align*}
    and thus equation (\ref{boundp}) becomes

    \begin{align*}
        \Pr_{x\sim p-2^{-\frac{\epsilon^2 n}{100}}}\Big[ \exists c\in C:x\succ c\Big]&\leq 2^{-\frac{\epsilon^2 n}{100}-1}+e^{-\frac{\epsilon n}{3}}\\
        &\leq 2^{-\frac{\epsilon^2 n}{100}}.
    \end{align*}
    But by definition $p=p^*-\frac{2\epsilon}{3}\cdot \frac{n}{n-1}$ and by our theorem's conditions $2^{-\frac{\epsilon^2 n}{100}}<\frac{\epsilon}{4}$, so we get
    \begin{align*}
        \Pr_{x\sim p^*-\epsilon}\Big[ \exists c\in C:x\succ c\Big]\leq 2^{-\frac{\epsilon^2 n}{100}}.
    \end{align*}
\end{proof}

\begin{remark}
    The work of Tillich and Z{\'{e}}mor gives very sharp block error decays in terms of the minimum distance of a linear code \cite{tillich2000sharptransition}. Thus, the primary purpose here is to bound the distance in $p$ between the block-error and bit-error thresholds. Once we have a bound such as the one given by our Theorem \ref{rmcapacity}, we immediately get a strong error decay by \cite{tillich2000sharptransition}.
\end{remark}

\begin{remark}
For Reed--Muller codes specifically, the ratios $\{\frac{d_r}{r}\}$ are large enough that one does not need the full power of Theorem \ref{bitcapacity}'s bit-error decay, $\Pr_{x\sim p}[\textnormal{bit error}]\leq e^{-(p^*-p)n}$. To show that the bit-error and block-error thresholds of Reed--Muller codes are asymptotically equal, it would have been sufficient to use a bit error decay of $e^{-(p^*-p)a\sqrt{n\log n}}$ for any $a=\omega(1)$. This is because by Corollary \ref{ratiorm}, we have $\frac{d_r}{r}\geq n^{10}$ for all $r\leq 2^{n-10\sqrt{n\log n}}$. Setting $p=p^*-\frac{20\ln2}{a}$ and applying Theorem \ref{bittoblock}
would then have given that the block-error probability under noise $p^*-\frac{20\ln2}{a}-\frac{1}{n^4}$ is bounded by $\frac{2}{n^4}.$

\end{remark}
\section*{Acknowledgments}
We thank Iwan Duursma for useful comments. This work was initiated while the authors were visiting the Simons Institute for the Theory of Computing. The work of Henry D. Pfister was supported in part by NSF CCF-2106213. The work of Oscar Sprumont was supported in part by NSF CCF-2131899, NSF CCF-1813135 and Anna Karlin's Bill and Melinda Gates Endowed Chair. The work of Gilles Z{\'{e}}mor was supported in part by Plan France 2030 through the project NISQ2LSQ, ANR-22-PETQ-0006.

\bibliographystyle{alpha}
\bibliography{erasure}

\end{document}